\documentclass[english]{article}
\usepackage[T1]{fontenc}
\usepackage[latin9]{inputenc}
\usepackage{geometry}
\usepackage{color}
\usepackage{babel}
\usepackage{amsmath}
\usepackage{amsthm}
\usepackage{amssymb}
\usepackage[authoryear]{natbib}
\usepackage[unicode=true,pdfusetitle,
 bookmarks=true,bookmarksnumbered=false,bookmarksopen=false,
 breaklinks=false,pdfborder={0 0 0},pdfborderstyle={},backref=false,colorlinks=true]
 {hyperref}
\hypersetup{
 citecolor=red,linkcolor=blue}

\makeatletter
\theoremstyle{plain}
\newtheorem{thm}{\protect\theoremname}
  \theoremstyle{remark}
  \newtheorem{rem}{\protect\remarkname}
  \theoremstyle{plain}
  \newtheorem{cor}{\protect\corollaryname}
  \theoremstyle{definition}
  \newtheorem*{example*}{\protect\examplename}
  \theoremstyle{plain}
  \newtheorem{prop}{\protect\propositionname}
  \theoremstyle{plain}
  \newtheorem{lem}{\protect\lemmaname}

\makeatother

  \providecommand{\examplename}{Example}
  \providecommand{\lemmaname}{Lemma}
  \providecommand{\propositionname}{Proposition}
  \providecommand{\remarkname}{Remark}
\providecommand{\corollaryname}{Corollary}
\providecommand{\theoremname}{Theorem}

\begin{document}

\title{Which ergodic averages have finite asymptotic variance?}

\author{George Deligiannidis$^{\star,\ddagger}$ and Anthony Lee$^{\dagger,\ddagger}$\\
\\
$^{\star}$University of Oxford, $^{\dagger}$University of Warwick,
$^{\ddagger}$Alan Turing Institute.}
\maketitle
\begin{abstract}
We show that the class of $L^{2}$ functions for which ergodic averages
of a reversible Markov chain have finite asymptotic variance is determined
by the class of $L^{2}$ functions for which ergodic averages of its
associated jump chain have finite asymptotic variance. This allows
us to characterize completely which ergodic averages have finite asymptotic
variance when the Markov chain is an independence sampler. In addition,
we obtain a simple sufficient condition for all ergodic averages of
$L^{2}$ functions of the primary variable in a pseudo-marginal Markov
chain to have finite asymptotic variance.
\end{abstract}

\section{Introduction}

On a measurable space $(\mathsf{E},\mathcal{E})$, let $\Phi:=(\Phi_{n})_{n\in\mathbb{N}}$
be an ergodic, reversible, discrete time Markov chain with Markov
transition kernel $\Pi$ and invariant probability measure $\mu$.
By ergodic, we mean $\Phi$ is $\mu$-irreducible. Such chains are
often simulated on a computer for the purpose of computing Monte Carlo
approximations of integrals $\mu(f):=\int_{\mathsf{E}}f(x)\mu({\rm d}x)$,
where $f\in L^{1}(\mathsf{E},\mu):=\{g:\mu(|g|)<\infty\}$. Ergodic
averages, $n^{-1}\sum_{i=1}^{n}f(\Phi_{i})$, associated with such
Markov chains converge almost surely as $n\to\infty$ to $\mu(f)$
for $\mu$-almost all $\Phi_{1}$ and all $f\in L^{1}(\mathsf{E},\mu)$
\citep[see, e.g.,][Chapter~17]{meyn2009markov}, and are frequently
used to approximate intractable integrals in computer science, physics
and statistics. The behaviour of such approximations is now quite
well understood, and central limit theorems (CLTs) for rescaled ergodic
averages and quantitative bounds on their asymptotic variance have
been established in a number of settings. We define the asymptotic
variance of ergodic averages of a function $f\in L^{2}(\mathsf{E},\mu):=\left\{ g:\mu(g^{2})<\infty\right\} $
to be
\[
{\rm var}(f,\Pi):=\lim_{n\rightarrow\infty}n{\rm var}\left\{ \frac{1}{n}\sum_{i=1}^{n}f(\Phi_{i})\right\} ,\qquad\Phi_{1}\sim\mu.
\]
For ergodic, $\mu$-reversible Markov chains, this limit exists for
all $f\in L^{2}(\mathsf{E},\mu)$ but may be infinite. Denoting the
function $x\mapsto f(x)-c$ by $f-c$, where $c$ is a constant, we
observe that ${\rm var}(f,\Pi)={\rm var}(f-c,\Pi)$, and so one can
restrict consideration to zero-mean functions $f\in L_{0}^{2}(\mathsf{E},\mu):=\{g\in L^{2}(\mathsf{E},\mu):\mu(g)=0\}$
without loss of generality.

A strong qualitative property of a Markov chain is that it is variance
bounding \citep{roberts2008variance}: if $\Phi$ is variance bounding
then it satisfies 
\[
\sup_{f\in L_{0}^{2}(\mathsf{E},\mu),{\rm var}_{\mu}(f)=1}{\rm var}(f,\Pi)<\infty,
\]
where ${\rm var}_{\mu}(f)$ is the variance of $f(\Phi_{1})$ when
$\Phi_{1}\sim\mu$. For reversible Markov chains, variance bounding
is closely related to geometric ergodicity and equivalent to finite
${\rm var}(f,\Pi)$ for all $f\in L^{2}(\mathsf{E},\mu)$. By \citet{kipnis1986central},
this implies a $\sqrt{n}$-CLT for all $f\in L^{2}(\mathsf{E},\mu)$
with limiting variance equal to the asymptotic variance, i.e. $n^{-1/2}\sum_{i=1}^{n}\left[f(\Phi_{i})-\mu(f)\right]$
converges weakly to a $N\left(0,{\rm var}(f,\Pi)\right)$ random variable
when $\Phi_{1}\sim\mu$. Hence, variance bounding provides some qualitative
assurance of the practicality of using ergodic averages as approximations
of $\mu(f)$ for all $f$ such that ${\rm var}_{\mu}(f)<\infty$.

Some Markov chains used in practice are ergodic and reversible but
not variance bounding, so for at least some $f\in L_{0}^{2}(\mathsf{E},\mu)$,
${\rm var}(f,\Pi)$ is not finite: the proof of Theorem 7 of \citet{roberts2008variance}
constructs one such $f$. On such occasions, it is beneficial to have
some guarantees on the subset of $L_{0}^{2}(\mathsf{E},\mu)$ whose
ergodic averages do have finite asymptotic variance. Relevant results
in this spirit include Theorems~4.1\textendash 4.5 of \citet{jarner2002polynomial},
Theorem~2 of \citet{jarner2007convergence} and Theorem~4.1 of \citet{bednorz2008regeneration},
involving the verification of Foster\textendash Lyapunov drift criteria
and/or regenerative properties of $\Phi$. We note, however, that
these results concern explicitly the existence of a $\sqrt{n}$-CLT
with finite limiting variance rather than finiteness of the asymptotic
variance.

In this paper, we consider the class of $\pi$-reversible, ergodic
Markov chains $X:=(X_{n})_{n\in\mathbb{N}}$ evolving on $\mathsf{E}$
whose Markov transition kernel is of the form
\begin{equation}
P(x,A):=\varrho(x)\tilde{P}(x,A)+[1-\varrho(x)]\mathbf{1}_{A}(x),\qquad A\in\mathcal{E},\label{eq:PPtilde}
\end{equation}
where $\tilde{P}$ is the Markov transition kernel of a reversible
Markov chain $\tilde{X}:=(\tilde{X}_{n})_{n\in\mathbb{N}}$, and $\varrho:\mathsf{E}\rightarrow(0,1]$.
Such chains arise frequently in statistical applications, Metropolis\textendash Hastings
chains being a notable example. We will refer to $\tilde{X}$ as the
jump chain associated with $X$ and $\tilde{P}$ the jump kernel associated
with $P$. The invariant probability measure $\tilde{\pi}$ associated
with $\tilde{X}$ is related to $\pi$ through $\varrho$ and defined
in Section~\ref{SEC:RELNPPTILDE}. Jump chains have been studied
by \citet{douc2011vanilla} and \citet{doucet2015efficient}, but
for different purposes than here. Denoting $x\mapsto f(x)/\varrho(x)$
by $f/\varrho$, our first main result is that for $f\in L_{0}^{2}(\mathsf{E},\pi)$,
${\rm var}(f,P)<\infty$ if and only if $f/\varrho\in L_{0}^{2}(\mathsf{E},\tilde{\pi})$
and ${\rm var}(f/\varrho,\tilde{P})<\infty$, extending a result by
\citet{doucet2015efficient}. This equivalence is interesting because
it allows us to infer that when $\tilde{P}$ is variance bounding,
then those functions $f$ such that $f/\varrho\in L_{0}^{2}(\mathsf{E},\tilde{\pi})$
are exactly the functions in $L_{0}^{2}(\mathsf{E},\pi)$ for which
${\rm var}(f,P)<\infty$.

We apply this result to independent Metropolis\textendash Hastings
(IMH) Markov chains as well as pseudo-marginal Markov chains. When
$P$ is an IMH kernel, we characterize the class of $\pi$-integrable
functions satisfying ${\rm var}(f,P)<\infty$. To the best of our
knowledge, this is the first result of this kind for independence
samplers. Pseudo-marginal Markov chains \citep{lin2000noisy,Beaumont2003,Andrieu2009}
are a Monte Carlo innovation that has received considerable recent
attention. When $P$ is a pseudo-marginal kernel, $X$ is a $\pi$-reversible
Markov chain evolving on $\mathsf{E}=\mathsf{X}\times\mathbb{R}_{+}$,
where $\pi$ admits as a marginal the invariant distribution of a
$\bar{\pi}$-reversible, ``marginal'' Markov chain $\bar{X}$ evolving
on $\mathsf{X}$. The extension of the state space accommodates the
inclusion of what can be viewed as a multiplicative noise variable,
and simulating $X$ is in many respects like simulating a noisy version
of $\bar{X}$. The noise introduced is of great practical importance:
computationally one can simulate $X$ in some cases where one cannot
simulate $\bar{X}$, while the properties of the noise variables introduced
affect in a variety of ways the behaviour of $X$ and associated ergodic
averages. A brief summary of relevant results in this active area
of research can be found in Section~\ref{SEC:PMMCS}. Our main application
of the result above is to provide a simple, sufficient condition for
all ergodic averages of functions $f(\cdot,u)=f_{X}\in L_{0}^{2}(\mathsf{X},\bar{\pi})$
to have ${\rm var}(f,P)<\infty$ when $\bar{X}$ is variance bounding.
This condition is both necessary and sufficient in some settings,
but not in general, and amounts to uniformly bounded second moments
of the noise variables. This complements existing results, and in
particular we do not make explicit assumptions about $\bar{X}$ beyond
assuming it is variance bounding. In contrast, previous sufficient
conditions when $X$ is not itself variance bounding have been found
when the marginal chain is strongly uniformly ergodic, or under fairly
specific assumptions on $\bar{X}$ \citep[Remark~15]{andrieu2014establishing}.

We close this section with some notation and definitions. $\mathbb{N}$
denotes the positive integers, $\mathbb{R}_{+}$ the non-negative
reals. For $\nu$ a measure on a measurable space $(\mathsf{S},\mathcal{S})$,
and $f$ a measurable function, we denote $\nu(f):=\int_{\mathsf{S}}f(x)\nu({\rm d}x)$.
We define $L^{2}(\mathsf{S},\nu)=\{f:\nu(f^{2})<\infty\}$ and $L_{0}^{2}(\mathsf{S},\nu):=\{f\in L^{2}(\mathsf{S},\nu):\nu(f)=0\}$.
Similarly, $L^{1}(\mathsf{S},\nu)=\{f:\nu(|f|)<\infty\}$ and $L_{0}^{1}(\mathsf{S},\nu):=\{f\in L^{1}(\mathsf{S},\nu):\nu(f)=0\}$.
For functions $f,g:\mathsf{S}\rightarrow\mathbb{R}$ we write $f\cdot g$
for the function $x\mapsto f(x)g(x)$ and when $g$ is strictly positive
$f/g$ for the function $x\mapsto f(x)/g(x)$. For a $\mu$-reversible
Markov kernel $\Pi$, we will say $\Pi$ is variance bounding when
its associated Markov chain is variance bounding. We write $\wedge$
and $\vee$ to denote $\min$ and $\max$, respectively. When we refer
to a Geometric distribution, we mean the distribution on $\mathbb{N}$.

Many of our results rely on Dirichlet forms and the variational definition
of the right spectral gap of a Markov operator. For a generic measurable
space $(\mathsf{S},\mathcal{S})$ and measure $\mu$, we denote by
$\left\langle \cdot,\cdot\right\rangle _{\mu}$ the inner product
on $L^{2}(\mathsf{S},\mu)$. We often rely on viewing a $\mu$-reversible
Markov kernel $\Pi$ as a self-adjoint operator on $L^{2}(\mathsf{S},\mu)$
or $L_{0}^{2}(\mathsf{S},\mu)$; this should always be clear from
the context. We define the Dirichlet form of a such a Markov kernel,
for $f\in L^{2}(\mathsf{S},\mu)$ as
\[
\mathcal{E}_{\Pi}(f):=\left\langle f,(I-\Pi)f\right\rangle _{\mu}=\frac{1}{2}\int_{\mathsf{S}}\mu({\rm d}x)\Pi(x,{\rm d}y)[f(y)-f(x)]^{2}.
\]
The right spectral gap of $\Pi$, as an operator on $L_{0}^{2}(\mathsf{S},\mu)$,
is then
\begin{equation}
{\rm Gap}(\Pi):=\inf_{f\in L_{0}^{2}(\mathsf{S},\mu),\left\langle f,f\right\rangle _{\mu}=1}\mathcal{E}_{\Pi}(f)=\inf_{f\in L^{2}(\mathsf{S},\mu),{\rm var}_{\mu}(f)>0}\frac{\mathcal{E}_{\Pi}(f)}{{\rm var}_{\mu}(f)},\label{eq:gapdefn}
\end{equation}
and from Theorem~14 of \citet{roberts2008variance}, $\Pi$ is variance
bounding if and only if ${\rm Gap}(\Pi)>0$.

\section{Relationship between $X$ and $\tilde{X}$\label{SEC:RELNPPTILDE}}

We describe briefly the relationship between the chain $X$ and its
associated jump chain $\tilde{X}$, following \citet{douc2011vanilla}.
Since $X$ is $\pi$-reversible, it is straightforward to establish
that $\tilde{X}$ is an ergodic, $\tilde{\pi}$-reversible Markov
chain, where 
\begin{equation}
\tilde{\pi}({\rm d}x)=\pi({\rm d}x)\varrho(x)/\pi(\varrho).\label{eq:tildepiexpr}
\end{equation}
We also observe that $\pi(f)=\pi(\varrho)\tilde{\pi}(f/\varrho)$.
One can construct a realization of $X$ from $\tilde{X}$ as follows.
First introduce random variables $(\tau_{n})_{n\in\mathbb{N}}$ such
that for each $n\in\mathbb{N}$, $\tau_{n}$ is conditionally independent
of all other random variables given $\tilde{X}_{n}$ with $\tau_{n}\mid\{\tilde{X}_{n}=x\}\sim{\rm Geometric}(\varrho(x))$.
By defining $S_{n}:=\inf\{k\geq1\::\:\sum_{i=1}^{k}\tau_{i}\geq n\}$
for $n\in\mathbb{N}$, one can verify that $(\tilde{X}_{S_{n}})_{n\in\mathbb{N}}$
is a realization of $X$ with initial state $\tilde{X}_{1}$.

Our first main result is the following, the converse part of which
is the novel addition to Proposition~2 of \citet{doucet2015efficient}.
The relation (\ref{eq:varfPiff}) may seem obvious. Indeed, if one
assumes that $\mathrm{var}(f,P)$ and $\mathrm{var}(f/\rho,\tilde{P})$
are both finite, then (\ref{eq:varfPiff}) follows from the representation
of $X$ in terms of $\tilde{X}$ and a careful application of the
Kipnis\textendash Varadhan CLT \citep{kipnis1986central}, as in the
proof of Proposition~2 in \citet{doucet2015efficient}. The main
difficulty lies in proving the first part of the theorem, where the
path-wise relation between $X$ and $\tilde{X}$ does not offer much
traction without further restrictive assumptions.
\begin{thm}
\label{THM:JUMPVARIFF}Let $f\in L_{0}^{2}(\mathsf{E},\pi)$. Then
${\rm var}(f,P)<\infty\iff f/\varrho\in L_{0}^{2}(\mathsf{E},\tilde{\pi})$
and ${\rm var}(f/\varrho,\tilde{P})<\infty$. Moreover,
\begin{equation}
{\rm var}(f,P)=\pi\left(f^{2}/\varrho\right)-\pi(f^{2})+\pi(\varrho){\rm var}(f/\varrho,\tilde{P}).\label{eq:varfPiff}
\end{equation}
\end{thm}
\begin{proof}
The direction $(\Rightarrow)$ and the expression for the variance
is Proposition~2 of \citet{doucet2015efficient}. We provide here
the proof of $(\Leftarrow)$. We recall the variational expression
for the asymptotic variance associated with a $\mu$-reversible Markov
kernel $\Pi$ suggested by \citet{caracciolo1990nonlocal}, discussed
in Section~4 of \citet{andrieu2014establishing},
\begin{equation}
{\rm var}(f,\Pi)=2\left[\sup_{g\in L^{2}(\mathsf{E},\mu)}2\left\langle f,g\right\rangle _{\mu}-\mathcal{E}_{\Pi}(g)\right]-\left\langle f,f\right\rangle _{\mu},\qquad f\in L_{0}^{2}(\mathsf{E},\mu).\label{eq:varrepavar}
\end{equation}
We observe from (\ref{eq:PPtilde}) that for $g\in L^{2}(\mathsf{E},\pi)$,
\begin{equation}
\mathcal{E}_{P}(g)=\frac{1}{2}\int_{\mathsf{E}}\pi({\rm d}x)\varrho(x)\tilde{P}(x,{\rm d}y)[g(y)-g(x)]^{2}=\pi(\varrho)\mathcal{E}_{\tilde{P}}(g).\label{eq:dirichletformreln}
\end{equation}
and that $\left\langle f/\varrho,g\right\rangle _{\tilde{\pi}}=\left\langle f,g\right\rangle _{\pi}/\pi(\varrho)$.
Let $f/\varrho\in L_{0}^{2}(\mathsf{E},\tilde{\pi}),$ which implies
$f\in L_{0}^{2}(\mathsf{E},\pi)$. Since $L^{2}(\mathsf{E},\pi)\subseteq L^{2}(\mathsf{E},\tilde{\pi})$,
and using (\ref{eq:dirichletformreln}), 
\begin{eqnarray*}
\sup_{g\in L_{0}^{2}(\mathsf{E},\tilde{\pi})}2\left\langle f/\varrho,g\right\rangle _{\tilde{\pi}}-\mathcal{E}_{\tilde{P}}(g) & \geq & \sup_{g\in L^{2}(\mathsf{E},\pi)}2\left\langle f/\varrho,g\right\rangle _{\tilde{\pi}}-\mathcal{E}_{\tilde{P}}(g)\\
 & = & \frac{1}{\pi(\varrho)}\left[\sup_{g\in L^{2}(\mathsf{E},\pi)}2\left\langle f,g\right\rangle _{\pi}-\mathcal{E}_{P}(g)\right].
\end{eqnarray*}
Combining this bound with the expressions for both ${\rm var}(f/\varrho,\tilde{P})$
and ${\rm var}(f,P)$ using (\ref{eq:varrepavar}), we obtain
\[
{\rm var}(f,P)\leq\pi(\varrho){\rm var}(f/\varrho,\tilde{P})+\pi(f^{2}/\varrho)-\pi(f^{2}),
\]
so $f/\varrho\in L_{0}^{2}(\mathsf{E},\tilde{\pi})$ and ${\rm var}(f/\varrho,\tilde{P})<\infty\Rightarrow{\rm var}(f,P)<\infty$.
\end{proof}
\begin{rem}
A different proof of Theorem~\ref{THM:JUMPVARIFF} can also be obtained
through the analysis of the multiplication operator $T:f\mapsto f/\rho$
between the Hilbert spaces $(L_{0}^{2}(\mathsf{E},\pi),\langle\cdot,\cdot\rangle_{1})$
and $(L_{0}^{2}(\mathsf{E},\tilde{\pi}),\langle\cdot,\cdot\rangle_{2})$,
where 
\begin{align*}
\langle f,g\rangle_{1} & :=\langle(I-P)^{-1/2}f,(I-P)^{-1/2}g\rangle_{\pi},\\
\langle f,g\rangle_{2} & :=\pi(\rho)\langle(I-\tilde{P})^{-1/2}f,(I-\tilde{P})^{-1/2}g\rangle_{\tilde{\pi}}.
\end{align*}
In the process of showing that $T$ is invertible and therefore proving
Theorem~1, one also obtains the interesting fact that $T$ as defined
is in fact an isometry, that is 
\[
\langle f,f\rangle_{1}=\langle f/\rho,f/\rho\rangle_{2}.
\]
This proves (\ref{eq:varfPiff}) directly, without requiring a careful
application of the CLT as was done in the proof of Proposition~2
in \citet{doucet2015efficient}.
\end{rem}
\begin{cor}
\label{cor:whichfunctions}If $\tilde{P}$ is variance bounding and
$f\in L_{0}^{2}(\mathsf{E},\pi)$, then ${\rm var}(f,P)<\infty\iff f/\varrho\in L_{0}^{2}(\mathsf{E},\tilde{\pi})$.
\end{cor}
The following example illustrates one way this result can be applied.
\begin{example*}
Let $p<1/2$ and $\varrho:\mathbb{N}\rightarrow(0,1]$, and consider
the reversible Markov chain $X$ on $\mathbb{N}$ with $P(1,1)=1-\varrho(1)p$,
$P(1,2)=\varrho(1)p$ and for $x>1$, $P(x,x)=1-\varrho(x)$, $P(x,x+1)=\varrho(x)p$
and $P(x,x-1)=\varrho(x)(1-p)$. The jump chain $\tilde{X}$ is a
simple random walk on $\mathbb{N}$ with $\tilde{\pi}$ the Geometric$(1-p/[1-p])$
distribution, and since $p<1/2$ it is variance bounding \citep[see, e.g.,][Section~15.5.1]{meyn2009markov}.
We have $\pi(x)\propto[p/(1-p)]^{x}/\varrho(x)$ and it can be shown
that $X$ is variance bounding if and only if $\inf_{x\in\mathsf{X}}\varrho(x)>0$.
Irrespective of this, Corollary~\ref{cor:whichfunctions} implies
that the functions $f\in L_{0}^{2}(\mathbb{N},\pi)$ that have ${\rm var}(f,P)<\infty$
are those satisfying $\sum_{x\in\mathbb{N}}[p/(1-p)]^{x}f(x)^{2}/\varrho(x)^{2}<\infty$.
\end{example*}
The following Proposition states that $\tilde{P}$ inherits variance
bounding from $P$. The example above shows that the converse clearly
does not hold, and this is why Corollary~\ref{cor:whichfunctions}
provides a route to the characterization of functions whose ergodic
averages have finite asymptotic variance.
\begin{prop}
\label{prop:inheritance_vb}$P$ and $\tilde{P}$ satisfy ${\rm Gap}(\tilde{P})\geq{\rm Gap}(P)$.
\end{prop}
\begin{proof}
If ${\rm Gap}(P)=0$ then the result is trivial. If ${\rm Gap}(P)>0$,
then $\varrho^{*}:=\pi-{\rm ess}\inf_{x\in\mathsf{E}}\varrho(x)>0$
by Theorem~1 of \citet{Leea}. It follows that $L^{2}(\mathsf{E},\tilde{\pi})=L^{2}(\mathsf{E},\pi)$.
For $g\in L_{0}^{2}(\mathsf{E},\pi)$,
\[
\frac{{\rm var}_{\pi}(g)}{{\rm var}_{\tilde{\pi}}(g)}=\frac{\pi(g^{2})}{\tilde{\pi}(g^{2})-\tilde{\pi}(g)^{2}}\geq\frac{\pi(g^{2})}{\tilde{\pi}(g^{2})}=\frac{\pi(g^{2})\pi(\varrho)}{\pi(\varrho\cdot g^{2})}\geq\pi(\varrho),
\]
and so, for any $g\in L^{2}(\mathsf{E},\tilde{\pi})$ using (\ref{eq:dirichletformreln}),
\[
\frac{\mathcal{E}_{\tilde{P}}(g)}{{\rm var}_{\tilde{\pi}}(g)}=\frac{\mathcal{E}_{P}(g)/\pi(\varrho)}{{\rm var}_{\pi}(g)}\cdot\frac{{\rm var}_{\pi}(g)}{{\rm var}_{\tilde{\pi}}(g)}\geq\frac{\mathcal{E}_{P}(g)}{{\rm var}_{\pi}(g)}\geq{\rm Gap}(P),
\]
and it follows from (\ref{eq:gapdefn}) that ${\rm Gap}(\tilde{P})\geq{\rm Gap}(P)$.
\end{proof}
In the sequel we will apply Theorem~\ref{THM:JUMPVARIFF} exclusively
to the case where
\begin{equation}
P(x,A)=\int_{A}q(x,{\rm d}y)\alpha(x,y)+[1-\varrho(x)]\mathbf{1}_{A}(x),\qquad A\in\mathcal{E},\label{eq:Pgenqalpha}
\end{equation}
with $q$ a Markov kernel, $\alpha:\mathsf{E}^{2}\rightarrow[0,1]$
an acceptance probability function and $\varrho(x):=\int_{E}q(x,{\rm d}y)\alpha(x,y)$
denoting the probability of accepting a proposal from $q(x,\cdot)$.
In this case, the jump kernel $\tilde{P}$ is
\[
\tilde{P}(x,A)=\int_{A}q(x,{\rm d}y)\alpha(x,y)/\varrho(x),\qquad A\in\mathcal{E},
\]
and $\tilde{X}$ is the Markov chain of accepted proposals. A particular
$\alpha$, which guarantees $\pi$-reversibility of $P$, is the Metropolis\textendash Hastings
acceptance probability function \citep{Metropolis1953,Hastings1970}
\begin{equation}
\alpha(x,y)=1\wedge\frac{\pi({\rm d}y)q(y,{\rm d}x)}{\pi({\rm d}x)q(x,{\rm d}y)}.\label{eq:MHalpha}
\end{equation}

\section{Independent Metropolis\textendash Hastings\label{sec:Independent-Metropolis--Hastings}}

\subsection{Characterization of functions with finite asymptotic variance}

We now apply Theorem~\ref{THM:JUMPVARIFF} to characterize those
$f\in L_{0}^{2}(\mathsf{E},\pi)$ with finite ${\rm var}(f,P)$ when
$P$ is an IMH kernel. In fact, we are able to characterize those
$f\in L_{0}^{1}(\mathsf{E},\pi)$ with finite ${\rm var}(f,P)$ in
this specific case. An IMH kernel is a Metropolis\textendash Hastings
kernel where in (\ref{eq:Pgenqalpha}), $q(x,\cdot)=\mu(\cdot)$ for
all $x\in\mathsf{E}$, where $\mu$ is a probability measure such
that $\pi\ll\mu$. The acceptance probability (\ref{eq:MHalpha})
is 
\[
\alpha(x,y):=1\wedge\frac{w(y)}{w(x)},\qquad x,y\in\mathsf{E},\qquad\text{where}\quad w:=\frac{{\rm d}\pi}{{\rm d}\mu},
\]
The resulting IMH chain $X$ has been analyzed for various $\pi$
and $\mu$. For example, \citet{Tierney1994} noted that when $\bar{w}:=\pi-{\rm ess}\sup_{x\in\mathsf{E}}w(x)<\infty$,
$X$ is uniformly ergodic with a spectral gap of $1/\bar{w}$, and
\citet{mengersen1996rates} showed that when $\bar{w}=\infty$, $X$
is not even geometrically ergodic. In \citet{jarner2002polynomial}
and \citet{jarner2007convergence}, conditions guaranteeing polynomial
ergodicity of $X$ and hence finite associated asymptotic variances
for some functions are obtained under assumptions on $\pi$ and $\mu$.
Using Theorem~\ref{THM:JUMPVARIFF}, however, we are able to characterize
exactly the class of functions with finite associated asymptotic variances.
\begin{thm}
\label{thm:imh_theorem}Let $f\in L_{0}^{1}(\mathsf{E},\pi)$. For
the IMH, ${\rm var}(f,P)<\infty$ if and only if $f\in L_{0}^{2}(\mathsf{E},\pi)$
and $w\cdot f\in L_{0}^{2}(\mathsf{E},\mu)$.
\end{thm}
Lemma~\ref{lem:gen_bounds_ap} is used multiple times in our proofs.
\begin{lem}
\label{lem:gen_bounds_ap}Let $Y$ be a non-negative random variable
with $\mathbb{E}[Y]=1$. Then
\[
\frac{1}{\mathbb{E}[Y^{2}]+c}\leq\mathbb{E}\left[1\wedge\frac{Y}{c}\right]\leq1\wedge\frac{1}{c}.
\]
\end{lem}
\begin{proof}
For the upper bound, we have $\mathbb{E}\left[1\wedge\frac{Y}{c}\right]\leq1\wedge\mathbb{E}\left[\frac{Y}{c}\right]=1\wedge\frac{1}{c}$.
For the lower bound, letting $\nu$ be the probability measure associated
with $Y$,
\begin{eqnarray*}
\mathbb{E}\left[1\wedge\frac{Y}{c}\right] & = & \int_{\mathbb{R}_{+}}\nu({\rm d}y)\left(1\wedge\frac{y}{c}\right)=\int_{\mathbb{R}_{+}}\nu({\rm d}y)y\left(\frac{1}{y}\wedge\frac{1}{c}\right)\\
 & \geq & \left[\int_{\mathbb{R}_{+}}\nu({\rm d}y)y\left(y\vee c\right)\right]^{-1}\geq\left(\mathbb{E}[Y^{2}]+c\right)^{-1},
\end{eqnarray*}
where we have used the fact that $\nu({\rm d}y)y$ is also a probability
measure, Jensen's inequality and $a\vee b\leq a+b$.
\end{proof}
\begin{cor}
\label{cor:imh_rho_bounds}For the IMH,
\[
\frac{1}{\pi(w)+w(x)}\leq\varrho(x)\leq1\wedge\frac{1}{w(x)}.
\]
\end{cor}
\begin{lem}
\label{lem:posoperator_notL2_notokay}For the IMH, if $f\in L_{0}^{1}(\mathsf{E},\pi)\setminus L_{0}^{2}(\mathsf{E},\pi)$
then ${\rm var}(f,P)=\infty$.
\end{lem}
\begin{proof}
Let $A:=\{x\in\mathsf{E}:f(x)\geq0\}$ and $B:=\{x\in\mathsf{E}:f(x)\leq0\}$.
Since $f\in L_{0}^{1}(\mathsf{E},\pi)\setminus L_{0}^{2}(\mathsf{E},\pi)$,
at least one of $\pi(\mathbf{1}_{A}\cdot f^{2})$ or $\pi(\mathbf{1}_{B}\cdot f^{2})$
is infinite, so let $C$ be one of these and observe that $\mu(C)>0$
since $\pi\ll\mu$. We consider the event $(X_{1},\ldots,X_{n})\in C^{n}$,
noting that for $x\in C$, $P(x,C)\geq\mu(C)$. On the event $(X_{1},\ldots,X_{n})\in C^{n}$
we have $\left[\sum_{i=1}^{n}f(X_{i})\right]^{2}\geq f(X_{1})^{2}$,
and so for any $n\in\mathbb{N}$,
\begin{eqnarray*}
{\rm var}\left(\frac{1}{n}\sum_{i=1}^{n}f(X_{i})\right) & = & \mathbb{E}\left[\left\{ \frac{1}{n}\sum_{i=1}^{n}f(X_{i})\right\} ^{2}\right]\\
 & \geq & \mathbb{E}\left[\mathbf{1}_{C^{n}}(X_{1},\ldots,X_{n})n^{-2}f(X_{1})^{2}\right]\\
 & \geq & n^{-2}\mu(C)^{n-1}\pi(\mathbf{1}_{C}\cdot f^{2})=\infty.
\end{eqnarray*}
Hence ${\rm var}(f,P)$ is infinite. 
\end{proof}
\begin{lem}
\label{lem:imh_uniform_minor}For the IMH, $\tilde{P}$ satisfies
the one-step minorization condition
\[
\tilde{P}(x,A)\geq\pi(\varrho)\tilde{\pi}(A),\qquad x\in\mathsf{E},\qquad A\in\mathcal{E},
\]
so $\tilde{X}$ is uniformly ergodic. Therefore, for $f\in L_{0}^{2}(\mathsf{E},\pi)$,
${\rm var}(f,P)<\infty\iff f/\varrho\in L_{0}^{2}(\mathsf{E},\tilde{\pi})$.
\end{lem}
\begin{proof}
Straightforward calculations and (\ref{eq:tildepiexpr}) provide,
\begin{eqnarray*}
\tilde{P}(x,{\rm d}y) & = & \frac{\mu({\rm d}y)\alpha(x,y)}{\varrho(x)}=\frac{\mu({\rm d}y)\left[1\wedge\frac{w(y)}{w(x)}\right]}{\varrho(x)}\\
 & = & \frac{\pi({\rm d}y)\left[\frac{1}{w(y)}\wedge\frac{1}{w(x)}\right]}{\varrho(x)}=\frac{\tilde{\pi}({\rm d}y)\pi(\varrho)\left[\frac{1}{w(y)}\wedge\frac{1}{w(x)}\right]}{\varrho(x)\varrho(y)}\\
 & = & \frac{\tilde{\pi}({\rm d}y)\pi(\varrho)}{\left[w(y)\vee w(x)\right]\varrho(x)\varrho(y)}\geq\tilde{\pi}({\rm d}y)\pi(\varrho),
\end{eqnarray*}
where in the inequality we have used the fact that when $w(x)>w(y)$,
\[
\left[w(y)\vee w(x)\right]\varrho(x)\varrho(y)=w(x)\varrho(x)\varrho(y)\leq\varrho(y)\leq1,
\]
and when $w(x)\leq w(y)$,
\[
\left[w(y)\vee w(x)\right]\varrho(x)\varrho(y)=w(y)\varrho(x)\varrho(y)\leq\varrho(x)\leq1,
\]
both cases involving application of the upper bound in Corollary~\ref{cor:imh_rho_bounds}.
For the second part, since $\tilde{X}$ is uniformly ergodic it is
variance bounding, and the result follows from Corollary~\ref{cor:whichfunctions}.
\end{proof}
\begin{lem}
\label{lem:imh_functions}For the IMH, let $f\in L_{0}^{2}(\mathsf{E},\pi)$.
Then $f/\varrho\in L_{0}^{2}(\mathsf{E},\tilde{\pi})$ if and only
if $w\cdot f\in L_{0}^{2}(\mathsf{E},\mu)$.
\end{lem}
\begin{proof}
We note that $f/\varrho\in L_{0}^{2}(\mathsf{E},\tilde{\pi})\iff\pi(f^{2}/\varrho)<\infty$.
If $\pi(f^{2}/\varrho)<\infty$ then since $w(x)\leq1/\varrho(x)$
by Lemma~\ref{cor:imh_rho_bounds}, 
\[
\mu(w^{2}\cdot f^{2})=\pi(w\cdot f^{2})\leq\pi(f^{2}/\varrho)<\infty.
\]
For the converse, assume $\pi(f^{2})<\infty$ and $\mu(w^{2}\cdot f^{2})=\pi(w\cdot f^{2})<\infty$.
We consider two cases: $\pi(w)<\infty$ and $\pi(w)=\infty$. If $\pi(w)<\infty$,
then $\varrho(x)\geq1/[\pi(w)+w(x)]$ by Lemma~\ref{cor:imh_rho_bounds},
so
\[
\pi(f^{2}/\varrho)\leq\pi(w)\pi(f^{2})+\pi(w\cdot f^{2})<\infty.
\]
If $\pi(w)=\infty$, then for each $x\in\mathsf{X}$, we define the
region of certain acceptance $A_{x}:=\left\{ y:w(y)\geq w(x)\right\} $
and observe that
\[
\varrho(x)=\int_{\mathsf{E}}1\wedge\frac{w(y)}{w(x)}\mu({\rm d}y)=\mu(A_{x})+\pi(A_{x}^{\complement})/w(x).
\]
Since $\mu(w)=1$, $w$ is $\mu$-almost everywhere finite and thus
there exists a $C>0$ such that $B:=\{x:w(x)\geq C\}$ satisfies $\pi(B^{\complement})>0$.
Moreover, since $\pi(w)=\infty$, we must have $\pi(B)>0$, which
implies $\mu(B)>0$ since $\pi\ll\mu$. We observe that 
\[
x\in B\Rightarrow w(x)\geq C\Rightarrow A_{x}\subseteq B\Rightarrow\pi(A_{x}^{\complement})\geq\pi(B^{\complement})\Rightarrow\varrho(x)\geq\pi(B^{\complement})/w(x),
\]
and similarly,
\[
x\in B^{\complement}\Rightarrow w(x)<C\Rightarrow A_{x}\supseteq B\Rightarrow\mu(A_{x})\geq\mu(B)\Rightarrow\varrho(x)\geq\mu(B).
\]
Therefore,
\[
\pi(f^{2}/\varrho)=\pi(\mathbf{1}_{B}\cdot f^{2}/\varrho)+\pi(\mathbf{1}_{B^{\complement}}\cdot f^{2}/\varrho)\leq\frac{\pi(\mathbf{1}_{B}\cdot w\cdot f^{2})}{\pi(B^{\complement})}+\frac{\pi(\mathbf{1}_{B^{\complement}}\cdot f^{2})}{\mu(B)}<\infty.\qedhere
\]
\end{proof}
\begin{proof}[Proof of Theorem~\ref{thm:imh_theorem}]
This is a consequence of Lemmas~\ref{lem:posoperator_notL2_notokay},~\ref{lem:imh_uniform_minor}
and~\ref{lem:imh_functions}.
\end{proof}
\begin{rem}
The characterization of $L_{0}^{1}(\mathsf{E},\pi)$ functions for
which independence sampler ergodic averages have finite asymptotic
variance involved extending the $L_{0}^{2}(\mathsf{E},\pi)$ characterization
with a specific result for this case, Lemma~\ref{lem:posoperator_notL2_notokay}.
We are not aware of general results for reversible Markov chains ensuring
that ergodic averages of functions that are in $L_{0}^{1}(\mathsf{E},\pi)$
but not $L_{0}^{2}(\mathsf{E},\pi)$ do not have finite asymptotic
variance, which would allow the characterization of Theorem~\ref{THM:JUMPVARIFF}
to be extended.
\end{rem}

\subsection{Comparison with self-normalized importance sampling}

Self-normalized importance sampling is an alternative way to define
a Monte Carlo approximation of $\pi(f)$ using a sequence of independent
$\mu$-distributed random variables $(Z_{n})_{n\in\mathbb{N}}$. If
we define 
\[
\pi_{n}^{{\rm SNIS}}(f):=\frac{\sum_{i=1}^{n}w(Z_{i})f(Z_{i})}{\sum_{i=1}^{n}w(Z_{i})},\qquad n\in\mathbb{N},
\]
one obtains that $\sqrt{n}\left\{ \pi_{n}^{{\rm SNIS}}(f)-\pi(f)\right\} $
converges weakly to a $N(0,\pi(w\cdot\bar{f}^{2}))$ random variable
whenever $\pi(w\cdot\bar{f}^{2})<\infty$, where $\bar{f}=f-\pi(f)$.
Theorem~\ref{thm:imh_theorem} indicates that the class of $L_{0}^{1}(\mathsf{E},\pi)$
functions $f$ with finite ${\rm var}(f,P)$ is in general smaller
than those satisfying $w\cdot f\in L_{0}^{2}(\mathsf{E},\mu)$. In
particular, small values of $w$ are able to counterbalance large
values of $f$ in $\pi_{n}^{{\rm SNIS}}(f)$, while $\varrho\leq1$
prevents any such counterbalancing for the IMH. The following bounds
allow us to compare ${\rm var}(f,P)$ with the limiting variance in
the self-normalized importance sampling CLT: the former is always
larger than the latter.
\begin{prop}
\label{prop:imh_compare_snis}If $f\in L_{0}^{2}(\mathsf{E},\pi)$
and $w\cdot f\in L_{0}^{2}(\mathsf{E},\mu)$, we have
\[
2\pi(\varrho)\tilde{\pi}(f^{2}/\varrho^{2})-\pi(f^{2})\leq{\rm var}(f,P)\leq2\tilde{\pi}(f^{2}/\varrho^{2})-\pi(f^{2}).
\]
and ${\rm var}(f,P)\geq\pi(w\cdot f^{2})$.
\end{prop}
\begin{proof}
Proposition~3 and Remark~1 of \citet{doucet2015efficient} show
that for the IMH, $\tilde{P}$ is a positive operator on $L_{0}^{2}(\mathsf{E},\tilde{\pi})$
so ${\rm var}(f/\varrho,\tilde{P})\geq\tilde{\pi}(f^{2}/\varrho^{2})$.
Lemma~\ref{lem:imh_uniform_minor} implies that ${\rm Gap}(\tilde{P})\geq\pi(\varrho)$
and spectral considerations (see, e.g., Section 3.5 of \citealp{geyer1992practical},
based on \citealp{kipnis1986central}) give 
\[
\tilde{\pi}(f^{2}/\varrho^{2})\leq{\rm var}(f/\varrho,\tilde{P})\leq\frac{2-{\rm Gap}(\tilde{P})}{{\rm Gap}(\tilde{P})}{\rm var}_{\tilde{\pi}}(f/\varrho)=\frac{2-\pi(\varrho)}{\pi(\varrho)}\tilde{\pi}(f^{2}/\varrho^{2}).
\]
These inequalities, together with (\ref{eq:varfPiff}), implies the
first set of inequalities. The last inequality follows from Corollary~\ref{cor:imh_rho_bounds}
since
\[
2\pi(\varrho)\tilde{\pi}(f^{2}/\varrho^{2})-\pi(f^{2})\geq\pi(w\cdot f^{2})+\pi(f^{2}/\varrho)-\pi(f^{2})\geq\pi(w\cdot f^{2}).\qedhere
\]
\end{proof}
\begin{rem}
When $f\in L_{0}^{2}(\mathsf{E},\pi)$ and $\bar{w}=\sup_{x\in\mathsf{X}}w(x)<\infty$,
spectral considerations provide the bounds $\pi(f^{2})\leq{\rm var}(f,P)\leq(2\bar{w}-1)\pi(f^{2})$.
The upper bound can be smaller or larger than the upper bound in Proposition~\ref{prop:imh_compare_snis},
but the first lower bound of Proposition~\ref{prop:imh_compare_snis}
is always larger than $\pi(f^{2})$.
\end{rem}

\section{Pseudo-marginal Markov chains\label{SEC:PMMCS}}

We briefly motivate the construction of pseudo-marginal chains, following
the notation of \citet{andrieu2015}. Let $\bar{\pi}$ be a probability
measure on $(\mathsf{X},\mathcal{X})$, and $\bar{X}$ the $\bar{\pi}$-reversible
Metropolis\textendash Hastings chain with proposal kernel $q$ and
acceptance probability function $\bar{\alpha}(x,y):=1\wedge\bar{r}(x,y)$,
where
\[
\bar{r}(x,y):=\frac{\bar{\pi}({\rm d}y)\bar{q}(y,{\rm d}x)}{\bar{\pi}({\rm d}x)\bar{q}(x,{\rm d}y)},\qquad x,y\in\mathsf{X}.
\]
Letting $\bar{\pi}$ and $\bar{q}$ have densities, also denoted by
$\bar{\pi}$ and $\bar{q}$, w.r.t. some reference measure, an associated
pseudo-marginal Markov chain $X$ can be constructed when only unbiased,
non-negative estimates of $\bar{\pi}(x)$ are available for each $x\in\mathsf{X}$.
That is, there exists a collection of probability measures $\{Q_{x}:x\in\mathsf{X}\}$
on non-negative noise variables such that 
\begin{equation}
\int_{\mathbb{R}_{+}}uQ_{x}({\rm d}u)=1,\qquad x\in\mathsf{X},\label{eq:unbiasedQ}
\end{equation}
and so if $U\sim Q_{x}$ then $U\bar{\pi}(x)$ is a non-negative random
variable with expectation $\bar{\pi}(x)$. Defining the probability
measure on $(\mathsf{E},\mathcal{E})=(\mathsf{X}\times\mathbb{R}_{+},\mathcal{X}\times\mathcal{B}(\mathbb{R}_{+}))$,
\[
\pi({\rm d}x,{\rm d}u):=\bar{\pi}({\rm d}x)Q_{x}({\rm d}u)u,
\]
the pseudo-marginal chain $X$ is a $\pi$-reversible Metropolis\textendash Hastings
chain with proposal kernel $q(x,u;{\rm d}y,{\rm d}v):=\bar{q}(x,{\rm d}y)Q_{y}({\rm d}v)$,
and acceptance probability function $\alpha(x,u;y,v):=1\wedge r(x,u;y,v)$,
where
\[
r(x,u;y,v):=\bar{r}(x,y)\frac{v}{u}=\frac{v\bar{\pi}({\rm d}y)\bar{q}(y,{\rm d}x)}{u\bar{\pi}({\rm d}x)\bar{q}(x,{\rm d}y)},\qquad(x,u),(y,v)\in\mathsf{E}.
\]
From a computational perspective, this means that only variables representing
the unbiased estimates $u\bar{\pi}(x)$ and $v\bar{\pi}(y)$ of the
densities $\bar{\pi}(x)$ and $\bar{\pi}(y)$ are required to compute
$\alpha$. Since the ratio of these densities appears in $r$, unbiased
estimates of the density $\bar{\pi}$ up to a common, but unknown,
normalizing constant are also sufficient; we focus here without loss
of generality on the case (\ref{eq:unbiasedQ}) to simplify the presentation
of the results, rather than allowing the R.H.S. therein to be an arbitrary
constant $c>0$.

The influence of $\{Q_{x}:x\in\mathsf{X}\}$ on the behaviour of $X$
and associated ergodic averages has recently been the subject of intense
research. For example, it is known that if the noise variables $U\sim Q_{x}$
are not almost surely bounded for $\bar{\pi}$-almost all $x$ then
$X$ cannot be variance bounding, while if they are essentially \emph{uniformly}
bounded then $X$ ``inherits'' variance bounding from $\bar{X}$
\citep{Andrieu2009,andrieu2015}. In between these cases, which is
fairly common in statistical applications, the situation is more complex
and $X$ may or may not inherit variance bounding depending on $\bar{q}$
\citep[see, e.g.,][]{Leea,andrieu2015}.

A simple version of a result by \citet{andrieu2014establishing} is
the establishment of a partial order between different pseudo-marginal
chains with noise variable distributions related by averaging independent
realizations of each $x$-dependent noise variable a fixed number
of times, extending results in \citet{andrieu2015} on the convergence
of \emph{finite} asymptotic variances to their marginal counterparts
in this setting. The issue of which ergodic averages have finite asymptotic
variances when $X$ is not variance bounding, however, has been resolved
only in a few specific settings through sub-geometric drift and minorization
conditions \citep[Remark~11]{andrieu2014establishing}. In addition,
a result by \citet{bornn2014use} and its generalization by \citet{sherlock2016}
shows that the class of functions with finite asymptotic variance
cannot be enlarged by averaging in the manner just described.

The pseudo-marginal kernel $P$ described above can be written, for
$A\in\mathcal{E}$,
\[
P(x,u;A):=\int_{A}\bar{q}(x,{\rm d}y)Q_{y}({\rm d}v)\left\{ 1\wedge r(x,u;y,v)\right\} +[1-\varrho(x,u)]\mathbf{1}_{A}(x,u),
\]
where $\varrho(x,u):=\int_{\mathsf{E}}q(x,{\rm d}y)Q_{y}({\rm d}v)\alpha(x,u;y,v)$,
and the marginal kernel $\bar{P}$ can be written, for $A\in\mathcal{X}$,
\[
\bar{P}(x,A):=\int_{A}\bar{q}(x,{\rm d}y)\left\{ 1\wedge\bar{r}(x,y)\right\} +[1-\bar{\varrho}(x)]\mathbf{1}_{A}(x),
\]
where $\bar{\varrho}(x):=\int_{\mathsf{E}}q(x,{\rm d}y)\alpha(x,y)$.

Our results are most easily stated by making reference to the second
moments of the noise variables, so we define
\[
s(x):=\int_{\mathbb{R}_{+}}u^{2}Q_{x}({\rm d}u),\qquad x\in\mathsf{X},
\]
and $\bar{s}:=\bar{\pi}-{\rm ess}\sup_{x\in\mathsf{X}}s(x)$.

\subsection{Independent proposals}

Our first result is a complete characterization of the functions $f\in L_{0}^{1}(\mathsf{E},\pi)$
satisfying ${\rm var}(f,P)$ in the specific case where $P$ is also
an IMH, and is essentially a corollary of Theorem~\ref{thm:imh_theorem}.
\begin{prop}
\label{PROP:IMH_PM}Assume $\bar{q}(x,\cdot)=\bar{\mu}(\cdot)$ for
all $x\in\mathsf{X}$, and $f\in L_{0}^{1}(\mathsf{E},\pi)$. Then
${\rm var}(f,P)<\infty$ if and only if $f\in L_{0}^{2}(\mathsf{E},\pi)$
and 
\[
\int_{\mathsf{E}}u^{2}\frac{{\rm d}\bar{\pi}}{{\rm d}\bar{\mu}}(x)f(x,u)^{2}\bar{\pi}({\rm d}x)Q_{x}({\rm d}u)<\infty.
\]
\end{prop}
\begin{proof}
If we define $\mu({\rm d}x,{\rm d}u)=\bar{\mu}({\rm d}x)Q_{x}({\rm d}u)$
then $P$ is exactly the $\pi$-reversible IMH kernel with proposal
$\mu$ and in particular, 
\[
w(x,u)=\frac{{\rm d}\pi}{{\rm d}\mu}(x,u)=u\frac{{\rm d}\bar{\pi}}{{\rm d}\bar{\mu}}(x).
\]
Theorem~\ref{thm:imh_theorem} then implies that ${\rm var}(f,P)<\infty$
if and only if $f\in L_{0}^{2}(\mathsf{E},\pi)$ and $w\cdot f\in L_{0}^{2}(\mathsf{E},\mu)$,
and so the result follows from
\begin{eqnarray*}
\mu(w^{2}\cdot f^{2}) & = & \pi(w\cdot f^{2})=\int_{\mathsf{E}}u\frac{{\rm d}\bar{\pi}}{{\rm d}\bar{\mu}}(x)f(x,u)^{2}\pi({\rm d}x,{\rm d}u)\\
 & = & \int_{\mathsf{E}}u^{2}\frac{{\rm d}\bar{\pi}}{{\rm d}\bar{\mu}}(x)f(x,u)^{2}\bar{\pi}({\rm d}x)Q_{x}({\rm d}u).\qedhere
\end{eqnarray*}
\end{proof}
\begin{cor}
\label{cor:IMH_PM}Assume $\bar{q}(x,\cdot)=\bar{\mu}(\cdot)$ for
all $x\in\mathsf{X}$. If $f(\cdot,u)=f_{X}\in L_{0}^{2}(\mathsf{X},\bar{\pi})$
then ${\rm var}(f,P)<\infty$ if and only if
\[
\int_{\mathsf{X}}s(x)\frac{{\rm d}\bar{\pi}}{{\rm d}\bar{\mu}}(x)f_{X}(x)^{2}\bar{\pi}({\rm d}x)<\infty,
\]
and clearly ${\rm var}(f,P)<\infty$ if $\bar{s}<\infty$ and ${\rm d}\bar{\pi}/{\rm d}\bar{\mu}$
is bounded above.
\end{cor}
\begin{rem}
\label{rem:imh_eg_abc}It is possible that $\sup_{x\in\mathsf{X}}s(x)\frac{{\rm d}\bar{\pi}}{{\rm d}\bar{\mu}}(x)<\infty$
even though $\bar{s}=\infty$. For example, let $\bar{\pi}({\rm d}x)\propto h(x)p({\rm d}x)$
and $\bar{\mu}=p$, where $h:\mathsf{X}\rightarrow(0,1)$ and $p$
is a probability measure, and $Q_{x}(\{1/h(x)\})=h(x)=1-Q_{x}(\{0\})$.
Then one obtains $s(x)=h(x)^{-1}$ and $\frac{{\rm d}\bar{\pi}}{{\rm d}\bar{\mu}}(x)\propto h(x)$,
so $s(x)\frac{{\rm d}\bar{\pi}}{{\rm d}\bar{\mu}}(x)$ is a constant
for all $x\in\mathsf{X}$. This is equivalent to the IMH for a simple
approximate Bayesian computation model, where $p$ is the prior distribution
of the statistical parameter and $h(x)$ the probability of the observed
data when $x$ is the true parameter \citep{Tavare1997}.
\end{rem}
\begin{rem}
\label{rem:imh_mupi}If for some $C>0$, $C^{-1}\leq\frac{{\rm d}\bar{\pi}}{{\rm d}\bar{\mu}}(x)\leq C$
for all $x\in\mathsf{X}$, then $\bar{s}<\infty$ is both necessary
and sufficient for all ergodic averages of $L_{0}^{2}(\mathsf{X},\bar{\pi})$
functions to have finite asymptotic variance. Perhaps surprisingly,
using $\bar{\pi}$ as the proposal distribution can make the class
of functions with finite asymptotic variance smaller when $\bar{s}=\infty$:
in the example of Remark~\ref{rem:imh_eg_abc} we obtain that this
class is exactly $L_{0}^{2}(\mathsf{X},p)$. Under this same condition
we also observe that a necessary and sufficient condition for all
bounded functions $f_{X}$ to have finite asymptotic variance is $\int_{\mathsf{X}}s(x)\bar{\pi}({\rm d}x)<\infty$.
\end{rem}
That the class of functions whose ergodic averages have finite asymptotic
variance depends on the second moment function $s$ is entirely consistent
with results by \citet{bornn2014use} and \citet{sherlock2016}. In
particular, we can consider defining for $N\in\mathbb{N}$ a new collection
of induced probability measures $\{Q_{x}^{N}:x\in\mathsf{X}\}$ where
for each $x\in\mathsf{X}$, $U\sim Q_{x}^{N}$ is equal in distribution
to the average of $N$ independent $Q_{x}$-distributed random variables.
If we define $v(x)$ to be the variance of $U\sim Q_{x}$, we obtain
$s(x)=1+v(x)$ and the second moment function $s_{N}$ associated
with $\{Q_{x}^{N}:x\in\mathsf{X}\}$ satisfies $s_{N}(x)=1+v(x)/N$.
It is then clear that
\[
\int_{\mathsf{X}}s(x)\frac{{\rm d}\bar{\pi}}{{\rm d}\bar{\mu}}(x)f_{X}(x)^{2}\bar{\pi}({\rm d}x)<\infty\iff\int_{\mathsf{X}}s_{N}(x)\frac{{\rm d}\bar{\pi}}{{\rm d}\bar{\mu}}(x)f_{X}(x)^{2}\bar{\pi}({\rm d}x)<\infty.
\]

\subsection{An auxiliary Markov kernel}

The remainder of our results provide sufficient conditions for the
ergodic averages of a function in $L^{2}(\mathsf{E},\pi)$ to have
finite asymptotic variance. The proofs are based on a modification
of $P$ whose associated asymptotic variances are larger than or equal
to those associated with $P$ itself, so that the novel converse part
of Theorem~\ref{THM:JUMPVARIFF} can still be applied to obtain results.
Strictly for the purpose of analysis, as in \citet{andrieu2015} and
\citet{doucet2015efficient}, we introduce an auxiliary Markov kernel
$R$ that has the same proposal as $P$ but a different acceptance
probability function. In particular, the acceptance probability is
\[
\alpha_{R}(x,u;y,v):=\left\{ 1\wedge\bar{r}(x,y)\right\} \left\{ 1\wedge\frac{v}{u}\right\} .
\]
We can therefore write $R$ as 
\begin{equation}
R(x,u;A):=\int_{A}q(x,{\rm d}y)Q_{y}({\rm d}v)\bar{\alpha}(x,y)\left\{ 1\wedge\frac{v}{u}\right\} +[1-\varrho_{R}(x,u)]\mathbf{1}_{A}(x,u),\label{eq:R_gen}
\end{equation}
where
\begin{equation}
\varrho_{R}(x,u):=\int_{\mathsf{E}}q(x,{\rm d}y)Q_{y}({\rm d}v)\alpha_{R}(x,u;y,v).\label{eq:rhoRdependent}
\end{equation}
It is straightforward to deduce that $R$ is $\pi$-reversible, e.g.
by Lemma~2 of \citet{banterle2015accelerating}, and also that $\alpha_{R}(x,u;y,v)\leq\alpha(x,u;y,v)$
for all $(x,u),(y,v)\in\mathsf{E}$. $P$ and $R$ are therefore ordered
in the sense of Peskun \citep{Peskun1973,Tierney1998}, so ${\rm var}(f,P)\leq{\rm var}(f,R)$
for all $f\in L^{2}(\mathsf{E},\pi)$.

Lemma~\ref{lem:equal_gaps} below could be deduced from Proposition~8
of \citet{andrieu2015}, in which the context is slightly different.
We provide a proof for completeness.
\begin{lem}
\label{lem:equal_gaps}Let $\mu({\rm d}x,{\rm d}u)=\nu({\rm d}x)\mu_{x}({\rm d}u)$
be a measure on $(\mathsf{E},\mathcal{E})=(\mathsf{X}\times\mathsf{U},\mathcal{\mathcal{X}}\times\mathcal{U})$.
Let $Q$ be a $\nu$-reversible sub-Markov kernel on $(\mathsf{X},\mathcal{X})$,
$\varrho$ be the function $x\mapsto Q(x,\mathsf{X})$, and $\bar{P}$
be the $\nu$-reversible Markov kernel 
\[
\bar{P}(x,A)=\int_{A}Q(x,{\rm d}y)+[1-\varrho(x)]\mathbf{1}_{A}(x),\qquad A\in\mathcal{X}.
\]
Letting $P$ be the $\mu$-reversible kernel 
\[
P(x,u;A)=\int_{A}Q(x,{\rm d}y)\mu_{y}({\rm d}v)+[1-\varrho(x)]\mathbf{1}_{A}(x,u),\qquad A\in\mathcal{E},
\]
we have ${\rm Gap}(\bar{P})\wedge\varrho^{*}\leq{\rm Gap}(P)\leq{\rm Gap}(\bar{P})$,
where $\varrho^{*}=\nu-{\rm ess}\inf_{x\in\mathsf{X}}\varrho(x)$.
\end{lem}
\begin{proof}
Let $f\in L_{0}^{2}(\mathsf{E},\mu)$ with $\left\langle f,f\right\rangle _{\mu}=1$.
For each $x\in\mathsf{X}$, we write $f_{x}$ for the function $u\mapsto f(x,u)$.
Let $\bar{f}(x):=\mu_{x}(f_{x})=\int_{\mathsf{U}}f(x,u)\mu_{x}({\rm d}u)$
and note that $\bar{f}\in L_{0}^{2}(\mathsf{X},\nu)$. When a function
$g\in L^{2}(\mathsf{X},\nu)$ is treated as a function in $L^{2}(\mathsf{E},\mu)$,
we adopt the convention that $g(\cdot,u)=g$. We observe that
\[
\mathcal{E}_{P}(f)=\left\langle \varrho\cdot f,f\right\rangle _{\mu}-\left\langle \varrho\cdot\bar{f},\bar{f}\right\rangle _{\nu}+\mathcal{E}_{\bar{P}}(\bar{f}).
\]
Let $h(x)={\rm var}_{\mu_{x}}(f_{x})$. Then for any $g\in L^{2}(\mathsf{X},\nu)$
we have $\left\langle g\cdot f,f\right\rangle _{\mu}-\left\langle g\cdot\bar{f},\bar{f}\right\rangle _{\nu}=\left\langle g,h\right\rangle _{\nu}$,
and so 
\begin{eqnarray*}
\mathcal{E}_{P}(f) & = & \left\langle \varrho,h\right\rangle _{\nu}+\mathcal{E}_{\bar{P}}(\bar{f})\geq\left\langle \varrho,h\right\rangle _{\nu}+\left\langle \bar{f},\bar{f}\right\rangle _{\nu}{\rm Gap}(\bar{P})\\
 & \geq & \left\langle \varrho,h\right\rangle _{\nu}+\left\langle \bar{f},\bar{f}\right\rangle _{\nu}{\rm Gap}(\bar{P})\wedge\varrho^{*}\\
 & = & \left\langle \varrho,h\right\rangle _{\nu}+\left\{ \left\langle f,f\right\rangle _{\mu}-\left\langle 1,h\right\rangle _{\nu}\right\} {\rm Gap}(\bar{P})\wedge\varrho^{*}\\
 & = & {\rm Gap}(\bar{P})\wedge\varrho^{*}+\left\langle \varrho,h\right\rangle _{\nu}-\left\langle {\rm Gap}(\bar{P})\wedge\varrho^{*},h\right\rangle _{\nu}\\
 & \geq & {\rm Gap}(\bar{P})\wedge\varrho^{*}.
\end{eqnarray*}
Since $f\in L_{0}^{2}(\mathsf{E},\mu)$ is arbitrary with $\left\langle f,f\right\rangle _{\mu}=1$,
we obtain from (\ref{eq:gapdefn}) that ${\rm Gap}(P)\geq{\rm Gap}(\bar{P})\wedge\varrho^{*}$.
That ${\rm Gap}(P)\leq{\rm Gap}(\bar{P})$ also follows from (\ref{eq:gapdefn})
by considering functions $f$ of $x$ alone in $L_{0}^{2}(\mathsf{E},\mu)$,
since then $\mathcal{E}_{P}(f)=\mathcal{E}_{\bar{P}}(\bar{f})$.
\end{proof}

\subsection{Independent noise distributions\label{subsec:Independent-noise-distributions}}

Our first result assumes that the noise distribution is state-independent,
i.e. $Q_{x}=Q$ for all $x\in\mathsf{X}$, and that the marginal jump
chain is variance bounding.
\begin{prop}
\label{PROP:INDEP_NOISE}Assume $Q_{x}=Q$ for all $x\in\mathsf{X}$,
$\bar{s}<\infty$, and that the jump kernel associated with $\bar{P}$
is variance bounding. Then,

\begin{enumerate}
\item For $f\in L_{0}^{2}(\mathsf{E},\pi)$, ${\rm var}(f,P)<\infty$ if
\[
\int_{\mathsf{E}}(u+\bar{s})u\frac{f(x,u)^{2}}{\bar{\varrho}(x)}\bar{\pi}({\rm d}x)Q({\rm d}u)<\infty.
\]
\item If $f(\cdot,u)=f_{X}\in L_{0}^{2}(\mathsf{X},\bar{\pi})$, then ${\rm var}(f_{X},\bar{P})<\infty\Rightarrow{\rm var}(f,P)<\infty$.
\end{enumerate}
\end{prop}
These results complement the analyses by \citet{doucet2015efficient}
and \citet{sherlock2015}, who assume that the distribution of the
weights is independent of $x$ in order to optimize the trade-off
between computational cost and asymptotic variance. In particular,
Proposition~\ref{PROP:INDEP_NOISE} indicates that those results
can be applied to ergodic averages of all $L_{0}^{2}(\mathsf{X},\bar{\pi})$
functions when the jump kernel associated with $\bar{P}$ is variance
bounding.

When the noise distribution is state-independent, (\ref{eq:R_gen})
simplifies to
\[
R(x,u;A)=\int_{A}\bar{q}(x,{\rm d}y)Q({\rm d}v)\left\{ 1\wedge r(x,y)\right\} \left\{ 1\wedge\frac{v}{u}\right\} +[1-\varrho_{R}(x,u)]\mathbf{1}_{A}(x,u),
\]
where $\varrho_{R}(x,u)=\int_{\mathsf{E}}q(x,{\rm d}y)Q({\rm d}v)\alpha_{R}(x,u;y,v)$.
If we define
\begin{equation}
\varrho_{U}(u):=\int_{\mathbb{R}_{+}}Q({\rm d}v)\left\{ 1\wedge\frac{v}{u}\right\} ,\qquad u\in\mathbb{R}_{+},\label{eq:rhoUindepnoise}
\end{equation}
then we observe that $\varrho_{R}(x,u)=\bar{\varrho}(x)\varrho_{U}(u)$. 
\begin{lem}
\label{lem:indep_weights_rhobounds}With $\bar{s}=\int_{\mathbb{R}_{+}}u^{2}Q({\rm d}u)$,
$\varrho_{U}$ in (\ref{eq:rhoUindepnoise}) satisfies 
\[
\frac{1}{\bar{s}+u}\leq\varrho_{U}(u)\leq1\wedge\frac{1}{u},
\]
and $\int_{\mathbb{R}_{+}}Q({\rm d}u)u\varrho_{U}(u)\geq(2\bar{s})^{-1}$.
\end{lem}
\begin{proof}
The first part follows from Lemma~\ref{lem:gen_bounds_ap}, since
$\varrho_{U}(u)=\mathbb{E}\left[1\wedge\frac{V}{u}\right]$, where
$V\sim Q$, and $V$ is a non-negative random variable with expectation
$1$. The second part follows from the first part and Jensen's inequality.
\end{proof}
\begin{proof}[Proof of Proposition~\ref{PROP:INDEP_NOISE}]
Let $\tilde{R}$ be the jump Markov kernel associated with $R$,
i.e. 
\[
\tilde{R}(x,u;{\rm d}y,{\rm d}v):=\frac{\bar{q}(x,{\rm d}y)Q({\rm d}v)\left\{ 1\wedge\bar{r}(x,y)\right\} \left\{ 1\wedge\frac{v}{u}\right\} }{\varrho_{R}(x,u)}.
\]
From (\ref{eq:tildepiexpr}) and $\varrho_{R}(x,u)=\bar{\varrho}(x)\varrho_{U}(u)$,
$\tilde{R}$ is $\mu$-reversible where, with $Id$ the identity function,
\[
\mu({\rm d}x,{\rm d}u):=\frac{\bar{\pi}({\rm d}x)\bar{\varrho}(x)}{\bar{\pi}(\bar{\varrho})}\frac{Q({\rm d}u)u\varrho_{U}(u)}{Q(Id\cdot\varrho_{U})}.
\]
We introduce an auxiliary Markov kernel $M$ which is also $\mu$-reversible:
\[
M(x,u;A)=\int_{A}\frac{\bar{q}(x,{\rm d}y)\left\{ 1\wedge\bar{r}(x,y)\right\} }{\bar{\varrho}(x)}\frac{Q({\rm d}v)v\varrho_{U}(v)}{Q(Id\cdot\varrho_{U})},\qquad A\in\mathcal{E}.
\]
For clarity, denote by $P^{*}$ the jump kernel associated with $\bar{P}$.
The strategy of the proof is to show that ${\rm Gap}(P^{*})>0\Rightarrow{\rm Gap}(M)>0\Rightarrow{\rm Gap}(\tilde{R})>0$,
and then to identify which functions $f$ satisfy $f/\varrho_{R}\in L_{0}^{2}(\mathsf{E},\mu)$,
since then ${\rm var}(f,P)\leq{\rm var}(f,R)<\infty$. We observe
that $M$ defines a Markov chain in which the first coordinate evolves
according to $P^{*}$, and the second coordinate is a sequence of
i.i.d. random variables. Hence, ${\rm Gap}(P^{*})>0\Rightarrow{\rm Gap}(M)>0$
by applying Lemma~\ref{lem:equal_gaps} with $Q=P^{*}$. We have
\begin{eqnarray*}
\tilde{R}(x,u;{\rm d}y,{\rm d}v) & = & \frac{\bar{q}(x,{\rm d}y)\left\{ 1\wedge\bar{r}(x,y)\right\} }{\bar{\varrho}(x)}\frac{Q({\rm d}v)v\left\{ \frac{1}{v}\wedge\frac{1}{u}\right\} }{\varrho_{U}(u)}\\
 & = & M(x,u;{\rm d}y,{\rm d}v)\frac{1}{\varrho_{U}(u)\varrho_{U}(v)[v\vee u]}Q(Id\cdot\varrho_{U}).
\end{eqnarray*}
From Lemma~\ref{lem:indep_weights_rhobounds}, we have $\varrho_{U}(u)\leq1\wedge1/u$,
so that $\varrho_{U}(u)\varrho_{U}(v)[v\vee u]\leq1$ by the same
argument as in the proof of Lemma~\ref{lem:imh_uniform_minor}. Hence,
$\tilde{R}(x,u;{\rm d}y,{\rm d}v)\geq Q(Id\cdot\varrho_{U})M(x,u;{\rm d}y,{\rm d}v)$,
and it follows that
\[
0<\mathcal{E}_{M}(f)\leq Q(Id\cdot\varrho_{U})^{-1}\mathcal{E}_{\tilde{R}}(f),\qquad f\in L_{0}^{2}(\mathsf{E},\mu),
\]
and so ${\rm Gap}(\tilde{R})\geq Q(Id\cdot\varrho_{U}){\rm Gap}(M)>0$
since $Q(Id\cdot\varrho_{U})\geq(2\bar{s})^{-1}$ by Lemma~\ref{lem:indep_weights_rhobounds}.
Since ${\rm var}(f,R)\geq{\rm var}(f,P)$, application of Corollary~\ref{cor:whichfunctions}
provides that all $f\in L_{0}^{2}(\mathsf{E},\pi)$ satisfying $f/\varrho_{R}\in L_{0}^{2}(\mathsf{E},\mu)$
have ${\rm var}(f,P)<\infty$, and we conclude the first part by writing
\begin{eqnarray*}
\pi(\varrho_{R})\mu(f^{2}/\varrho_{R}^{2}) & = & \int_{\mathsf{E}}\frac{f(x,u)^{2}}{\varrho_{R}(x,u)}\pi({\rm d}x,{\rm d}u)=\int_{\mathsf{E}}\frac{f(x,u)^{2}}{\bar{\varrho}(x)\varrho_{U}(u)}\bar{\pi}({\rm d}x)Q({\rm d}u)u\\
 & \leq & \int_{\mathsf{E}}(u+\bar{s})u\frac{f(x,u)^{2}}{\bar{\varrho}(x)}\bar{\pi}({\rm d}x)Q({\rm d}u),
\end{eqnarray*}
where the inequality follows from Lemma~\ref{lem:indep_weights_rhobounds}.
For the second part, we have
\begin{eqnarray*}
\pi(\varrho_{R})\mu(f^{2}/\varrho_{R}^{2}) & \leq & \int_{\mathsf{X}}\frac{f_{X}(x)^{2}}{\bar{\varrho}(x)}\bar{\pi}({\rm d}x)\int_{\mathbb{R}^{+}}(u+\bar{s})uQ({\rm d}u)\\
 & \leq & 2\bar{s}\int_{\mathsf{X}}\frac{f_{X}(x)^{2}}{\bar{\varrho}(x)}\bar{\pi}({\rm d}x),
\end{eqnarray*}
and $\bar{\pi}(f_{X}^{2}/\bar{\varrho})<\infty$ is equivalent by
Corollary~\ref{cor:whichfunctions} to ${\rm var}(f_{X},\bar{P})<\infty$
since $P^{*}$ is variance bounding.
\end{proof}

\subsection{General case}

Our most generally applicable result for pseudo-marginal chains is
the following. The strategy of the proof is similar in many respects
to that of Proposition~\ref{PROP:INDEP_NOISE}, but more complicated.
In addition, the assumption that $\bar{P}$ is variance bounding is
stronger (cf. Proposition~\ref{prop:inheritance_vb}) than the assumption
that its associated jump kernel is variance bounding.
\begin{thm}
\label{THM:GENPMSECOND}Assume $\bar{P}$ is variance bounding and
$\bar{s}<\infty$. Then for $f\in L_{0}^{2}(\mathsf{E},\pi)$ satisfying
\[
\int_{\mathsf{E}}f(x,u)^{2}\bar{\pi}({\rm d}x)Q_{x}({\rm d}u)u^{2}<\infty,
\]
${\rm var}(f,P)<\infty$. In particular, if $f(\cdot,u)=f_{X}\in L_{0}^{2}(\mathsf{X},\bar{\pi})$
then ${\rm var}(f,P)<\infty$.
\end{thm}
Remark~\ref{rem:imh_mupi} indicates that the condition $\bar{s}<\infty$
is also necessary in some settings, while of course Remark~\ref{rem:imh_eg_abc}
indicates that it is not necessary in others. 

In this case, $\varrho_{R}$ does not factorize as in Section~\ref{subsec:Independent-noise-distributions}
since the distribution of the weights is dependent on the proposed
value of $y$.
\begin{lem}
\label{lem:dep_weights_rhobounds}Let $\varrho_{R}$ be given by (\ref{eq:rhoRdependent}),
and $\varrho_{R,X}(x):=\int Q_{x}({\rm d}u)u\varrho_{R}(x,u)$. Then
for each $(x,u)\in\mathsf{E}$,
\[
\frac{\bar{\varrho}(x)}{\bar{s}+u}\leq\frac{\bar{\varrho}(x)}{s(y)+u}\leq\varrho_{R}(x,u)\leq\bar{\varrho}(x)\left\{ 1\wedge\frac{1}{u}\right\} ,
\]
and for each $x\in\mathsf{X}$, $\bar{\varrho}(x)/(2\bar{s})\leq\varrho_{R,X}(x)\leq\bar{\varrho}(x)$.
\end{lem}
\begin{proof}
We can write $\varrho_{R}(x,u)=\int_{\mathsf{X}}q(x,{\rm d}y)\bar{\alpha}(x,y)\int_{\mathbb{R}_{+}}Q_{y}({\rm d}v)\left[1\wedge\frac{v}{u}\right]$,
whence the first part holds by applying Lemma~\ref{lem:gen_bounds_ap}
to the inner integral. For the second part, we have
\begin{eqnarray*}
\varrho_{R,X}(x) & = & \int_{\mathbb{R}_{+}}Q_{x}({\rm d}u)u\varrho_{R}(x,u)\\
 & = & \int_{\mathsf{X}}\bar{q}(x,{\rm d}y)\bar{\alpha}(x,y)\int_{\mathbb{R}_{+}^{2}}Q_{x}({\rm d}u)uQ_{y}({\rm d}v)\left[1\wedge\frac{v}{u}\right],
\end{eqnarray*}
so that $\varrho_{R,X}(x)\leq\varrho_{X}(x)$. Moreover,
\begin{eqnarray*}
\varrho_{R,X}(x) & = & \int_{\mathsf{X}}\bar{q}(x,{\rm d}y)\bar{\alpha}(x,y)\int_{\mathbb{R}_{+}^{2}}Q_{x}({\rm d}u)uQ_{y}({\rm d}v)v\left[\frac{1}{v}\wedge\frac{1}{u}\right]\\
 & \geq & \int_{\mathsf{X}}\bar{q}(x,{\rm d}y)\bar{\alpha}(x,y)/\left[s(x)+s(y)\right]\geq\int_{\mathsf{X}}\bar{q}(x,{\rm d}y)\bar{\alpha}(x,y)/2\bar{s},
\end{eqnarray*}
where we have used Jensen's inequality and the fact that $a\vee b\leq a+b$.
\end{proof}
\begin{proof}[Proof of Theorem~\ref{THM:GENPMSECOND}]
Let $\tilde{R}$ be the jump kernel associated with $R$, which from
(\ref{eq:tildepiexpr}) is $\mu$-reversible with 
\[
\mu({\rm d}x,{\rm d}u)=\frac{\bar{\pi}({\rm d}x)Q_{x}({\rm d}u)u\varrho_{R}(x,u)}{\pi(\varrho_{R})}.
\]
We decompose $\mu$ as $\mu({\rm d}x,{\rm d}u)=\nu({\rm d}x)\mu_{x}({\rm d}u)$
where $\nu({\rm d}x):=\bar{\pi}({\rm d}x)\varrho_{R,X}(x)/\pi(\varrho_{R})$
and $\mu_{x}({\rm d}u):=Q_{x}({\rm d}u)u\varrho_{R}(x,u)/\varrho_{R,X}(x)$.
We introduce a $\nu$-reversible, Markov kernel 
\[
\bar{M}(x,A):=\int_{A}\bar{q}(x,{\rm d}y)\bar{\alpha}(x,y)\left[1\wedge\frac{\varrho_{R,X}(y)}{\varrho_{R,X}(x)}\right]+[1-\varrho_{M}(x)]\mathbf{1}_{A}(x),\qquad A\in\mathcal{X},
\]
where $\varrho_{M}(x):=\int_{\mathsf{X}}\bar{q}(x,{\rm d}y)\bar{\alpha}(x,y)\left[1\wedge\varrho_{R,X}(y)/\varrho_{R,X}(x)\right]$.
We also introduce a $\mu$-reversible Markov kernel $M$, where for
$A\in\mathcal{E}$,
\[
M(x,u;A)=\int_{A}\bar{q}(x,{\rm d}y)\bar{\alpha}(x,y)\left[1\wedge\frac{\varrho_{R,X}(y)}{\varrho_{R,X}(x)}\right]\mu_{y}({\rm d}v)+[1-\varrho_{M}(x)]\mathbf{1}_{A}(x,u).
\]
The strategy of the proof is to show that ${\rm Gap}(\bar{P})>0\Rightarrow{\rm Gap}(\bar{M})>0\Rightarrow{\rm Gap}(M)>0\Rightarrow{\rm Gap}(\tilde{R})>0$
and then to identify which functions $f$ satisfy $f/\varrho_{R}\in L_{0}^{2}(\mathsf{E},\mu)$,
since then ${\rm var}(f,P)\leq{\rm var}(f,R)<\infty$. We observe
that $\bar{P}$ being variance bounding implies $\underline{\varrho}:=\inf_{x}\bar{\varrho}(x)>0$,
by Theorem~1 of \citet{Leea}. By Lemma~\ref{lem:dep_weights_rhobounds}
we have $\varrho_{R,X}(y)/\varrho_{R,X}(x)\geq\underline{\varrho}/(2\bar{s})$
so $\inf_{x\in\mathsf{X}}\varrho_{M}(x)\geq\underline{\varrho}/(2\bar{s})>0$.
By Lemma~\ref{lem:equal_gaps}, ${\rm Gap}(M)\geq{\rm Gap}(\bar{M})\wedge[\underline{\varrho}/(2\bar{s})]$
and we now show that ${\rm Gap}(\bar{M})>0$. Since $\underline{\varrho}>0$,
$L^{2}(\mathsf{X},\nu)=L^{2}(\mathsf{X},\bar{\pi})$. For $f\in L^{2}(\mathsf{X},\bar{\pi})$,
we have
\begin{eqnarray*}
2\mathcal{E}_{\bar{M}}(f) & = & \int_{\mathsf{X}}\frac{\bar{\pi}({\rm d}x)\varrho_{R,X}(x)}{\pi(\varrho_{R})}q(x,{\rm d}y)\bar{\alpha}(x,y)\left[1\wedge\frac{\varrho_{R,X}(y)}{\varrho_{R,X}(x)}\right]\left[f(y)-f(x)\right]^{2}\\
 & \geq & \frac{\underline{\varrho}}{2\bar{s}}\int\bar{\pi}({\rm d}x)q(x,{\rm d}y)\bar{\alpha}(x,y)\left[f(y)-f(x)\right]^{2}=2\frac{\underline{\varrho}}{2\bar{s}}\mathcal{E}_{\bar{P}}(f).
\end{eqnarray*}
Moreover, for $f\in L_{0}^{2}(\mathsf{X},\bar{\pi})$,
\[
\frac{{\rm var}_{\bar{\pi}}(f)}{{\rm var}_{\nu}(f)}=\frac{\pi(\varrho_{R})\bar{\pi}(f^{2})}{\bar{\pi}(\varrho_{R,X}\cdot f^{2})-\bar{\pi}(\varrho_{R,X}\cdot f)^{2}/\pi(\varrho_{R})}\geq\frac{\pi(\varrho_{R})\bar{\pi}(f^{2})}{\bar{\pi}(\varrho_{R,X}\cdot f^{2})}\geq\pi(\varrho_{R}),
\]
so that for all $f\in L^{2}(\mathsf{X},\nu)=L^{2}(\mathsf{X},\bar{\pi})$,
\[
\frac{\mathcal{E}_{\bar{M}}(f)}{{\rm var}_{\nu}(f)}\geq\frac{\underline{\varrho}}{2\bar{s}}\frac{\mathcal{E}_{\bar{P}}(f)}{{\rm var}_{\bar{\pi}}(f)}\cdot\frac{{\rm var}_{\bar{\pi}}(f)}{{\rm var}_{\nu}(f)}\geq\frac{\underline{\varrho}\pi(\varrho_{R})}{2\bar{s}}\frac{\mathcal{E}_{\bar{P}}(f)}{{\rm var}_{\bar{\pi}}(f)},
\]
and it follows from (\ref{eq:gapdefn}) that ${\rm Gap}(\bar{M})\geq{\rm Gap}(\bar{P})\underline{\varrho}\pi(\varrho_{R})/(2\bar{s})>0$,
and so ${\rm Gap}(M)>0$. Finally, we compare $\tilde{R}$ with $M$.
For $f\in L_{0}^{2}(\mathsf{E},\mu)$ we have 
\begin{eqnarray*}
2\mathcal{E}_{\tilde{R}}(f) & = & \int\mu({\rm d}x,{\rm d}u)\frac{q(x,{\rm d}y)Q_{y}({\rm d}v)\bar{\alpha}(x,y)\left[1\wedge\frac{v}{u}\right]}{\varrho_{R}(x,u)}\left[f(y,v)-f(x,u)\right]^{2}\\
 & = & \int\mu({\rm d}x,{\rm d}u)\frac{q(x,{\rm d}y)Q_{y}({\rm d}v)v\bar{\alpha}(x,y)\left[\frac{1}{v}\wedge\frac{1}{u}\right]}{\varrho_{R}(x,u)}\left[f(y,v)-f(x,u)\right]^{2}\\
 & = & \int\mu({\rm d}x,{\rm d}u)\frac{q(x,{\rm d}y)\bar{\alpha}(x,y)\left[\frac{1}{v}\wedge\frac{1}{u}\right]}{\varrho_{R}(x,u)}\frac{\mu_{y}({\rm d}v)\varrho_{R,X}(y)}{\varrho_{R}(y,v)}\left[f(y,v)-f(x,u)\right]^{2}.
\end{eqnarray*}
From Lemma~\ref{lem:dep_weights_rhobounds}, we know that $[v\vee u]\varrho_{R}(x,u)\varrho_{R}(y,v)\leq\bar{\varrho}(x)\bar{\varrho}(y)$
for all $(x,u),(y,v)\in\mathsf{E}$, so
\begin{eqnarray*}
2\mathcal{E}_{\tilde{R}}(f) & \geq & \int\mu({\rm d}x,{\rm d}u)q(x,{\rm d}y)\bar{\alpha}(x,y)\mu_{y}({\rm d}v)\frac{\varrho_{R,X}(y)}{\bar{\varrho}(x)\bar{\varrho}(y)}\left[f(y,v)-f(x,u)\right]^{2}\\
 & \geq & \frac{1}{2\bar{s}}\int\mu({\rm d}x,{\rm d}u)q(x,{\rm d}y)\bar{\alpha}(x,y)\mu_{y}({\rm d}v)\bar{\varrho}(x)^{-1}\left[f(y,v)-f(x,u)\right]^{2}\\
 & \geq & \frac{1}{2\bar{s}}\int\mu({\rm d}x,{\rm d}u)q(x,{\rm d}y)\bar{\alpha}(x,y)\mu_{y}({\rm d}v)\left[1\wedge\frac{\varrho_{R,X}(y)}{\varrho_{R,X}(x)}\right]\left[f(y,v)-f(x,u)\right]^{2}\\
 & = & \frac{1}{2\bar{s}}2\mathcal{E}_{M}(f).
\end{eqnarray*}

Hence ${\rm Gap}(\tilde{R})\geq\frac{1}{2\bar{s}}{\rm Gap}(M)>0$,
so by Corollary~\ref{cor:whichfunctions} all functions $f\in L_{0}^{2}(\mathsf{E},\pi)$
satisfying $f/\varrho_{R}\in L_{0}^{2}(\mathsf{E},\mu)$ have ${\rm var}(f,R)<\infty$.
We have
\begin{eqnarray*}
\pi(\varrho_{R})\int_{\mathsf{E}}\frac{f(x,u)^{2}}{\varrho_{R}(x,u)^{2}}\mu({\rm d}x,{\rm d}u) & = & \int_{\mathsf{E}}\frac{f(x,u)^{2}}{\varrho_{R}(x,u)}\bar{\pi}({\rm d}x)Q_{x}({\rm d}u)u\\
 & \leq & \frac{2\bar{s}}{\underline{\varrho}}\int_{\mathsf{E}}f(x,u)^{2}\bar{\pi}({\rm d}x)Q_{x}({\rm d}u)u^{2},
\end{eqnarray*}
and we conclude by noting that for $f\in L_{0}^{2}(\mathsf{E},\pi)$,
${\rm var}(f,P)\leq{\rm var}(f,R)$.
\end{proof}

\section{\label{SEC:RAO_BLACKWELL}On alternatives to geometric random variables}

One of the contributions of \citet{douc2011vanilla} is to consider
weighted ergodic averages associated with the Markov chain $\tilde{X}$
to estimate $\pi(f)$. In particular, they propose alternative random
weights to the $(\tau_{n})_{n\in\mathbb{N}}$ that ensure smaller
asymptotic variances of the estimators of $\pi(f)$. The purpose of
this last section is to point out that in many situations, the reduction
in variance can be limited.

We consider the sequence of estimators of $\pi(f)$, with $\tilde{X}_{1}\sim\tilde{\pi}$,
\[
\bar{Y}_{n}^{{\rm RB}}(f):=\frac{\sum_{i=1}^{n}f(\tilde{X}_{i})/\varrho(\tilde{X}_{i})}{\sum_{i=1}^{n}1/\varrho(\tilde{X}_{i})},\qquad\bar{Y}_{n}^{{\rm Geo}}(f):=\frac{\sum_{i=1}^{n}\tau_{i}f(\tilde{X}_{i})}{\sum_{i=1}^{n}\tau_{i}}\qquad n\geq1.
\]

\begin{prop}
\label{prop:avarrbgeosn}Let $f\in L^{1}(\pi)$, $\bar{f}:=f-\pi(f)$,
$\bar{f}/\varrho\in L_{0}^{2}(\mathsf{E},\tilde{\pi})$ and ${\rm var}(\bar{f}/\varrho,\tilde{P})<\infty$.
Then

\begin{enumerate}
\item $\bar{Y}_{n}^{{\rm RB}}(f)\overset{a.s.}{\rightarrow}\pi(f)$ and
$\bar{Y}_{n}^{{\rm Geo}}(f)\overset{a.s.}{\rightarrow}\pi(f)$ as
$n\rightarrow\infty$.
\item $\sqrt{n}\left[\bar{Y}_{n}^{{\rm RB}}(f)-\pi(f)\right]\overset{L}{\rightarrow}N(0,\sigma_{{\rm RB}}^{2}(f))$,
where 
\[
\sigma_{{\rm RB}}^{2}(f)=\pi(\varrho)^{2}{\rm var}(\tilde{P},\bar{f}/\varrho).
\]
\item $\sqrt{n}\left[\bar{Y}_{n}^{{\rm Geo}}(f)-\pi(f)\right]\overset{L}{\rightarrow}N(0,\sigma_{{\rm Geo}}^{2}(f))$,
where 
\[
\sigma_{{\rm Geo}}^{2}(f)=\pi(\varrho)\left\{ \pi(\bar{f}^{2}/\varrho)-\pi(\bar{f}^{2})+\pi(\varrho){\rm var}(\tilde{P},\bar{f}/\varrho)\right\} .
\]
\end{enumerate}
\end{prop}
\begin{rem}
The use of geometric random variables to construct the Markov chain
$X$ from the jump chain $\tilde{X}$ is responsible for the term
$\pi\left(f^{2}/\varrho\right)-\pi(f^{2})$ in $\sigma_{{\rm Geo}}^{2}(f)$.
We notice that when 
\begin{equation}
{\rm var}(\tilde{P},\bar{f}/\varrho)\geq\tilde{\pi}(\bar{f}^{2}/\varrho^{2})=\pi(\bar{f}^{2}/\varrho)/\pi(\varrho),\label{eq:Ptildenotamazing}
\end{equation}
then $\pi(\varrho)^{2}{\rm var}(\tilde{P},\bar{f}/\varrho)\geq\pi(\varrho)\pi(\bar{f}^{2}/\varrho)$
and so $\sigma_{{\rm Geo}}^{2}(f)\leq2\sigma_{{\rm RB}}^{2}(f)$.
We note that (\ref{eq:Ptildenotamazing}) holds, e.g., when $\tilde{P}$
is a positive operator on $L_{0}^{2}(\mathsf{E},\tilde{\pi})$.
\end{rem}
Hence, the computational benefits of the Rao\textendash Blackwellization
strategies are large only when the computational cost of obtaining
the improved estimates is considerably less than that of simulating
the chain itself.

In order to prove Proposition~\ref{prop:avarrbgeosn}, we first consider
the sequences of unbiased estimators, with $\tilde{X}_{1}\sim\tilde{\pi}$,
\[
Y_{n}^{{\rm RB}}(f):=\frac{\pi(\varrho)}{n}\sum_{i=1}^{n}\frac{f(\tilde{X}_{i})}{\varrho(\tilde{X}_{i})},\qquad Y_{n}^{{\rm Geo}}(f):=\frac{\pi(\varrho)}{n}\sum_{i=1}^{n}\tau_{i}f(\tilde{X}_{i}),\qquad n\geq1.
\]

\begin{lem}
\label{lem:avarRBGeo}Let $f/\varrho\in L^{2}(\mathsf{E},\tilde{\pi})$
and ${\rm var}(f/\varrho,\tilde{P})<\infty$. Then

\begin{enumerate}
\item $Y_{n}^{{\rm RB}}(f)\overset{a.s.}{\rightarrow}\pi(f)$ as $n\rightarrow\infty$
and $Y_{n}^{{\rm Geo}}(f)\overset{a.s.}{\rightarrow}\pi(f)$ as $n\rightarrow\infty$.
\item Their asymptotic variances are 
\begin{equation}
\lim_{n\rightarrow\infty}n{\rm var}(Y_{n}^{{\rm RB}}(f))=\pi(\varrho)^{2}{\rm var}(f/\varrho,\tilde{P}),\label{eq:avarjumppif}
\end{equation}
and 
\begin{equation}
\lim_{n\rightarrow\infty}n{\rm var}\left(Y_{n}^{{\rm Geo}}(f)\right)=\pi(\varrho)\pi\left(f^{2}/\varrho\right)-\pi(\varrho)\pi(f^{2})+\pi(\varrho)^{2}{\rm var}(f/\varrho,\tilde{P}).\label{eq:avarjumpGeo}
\end{equation}
\end{enumerate}
\end{lem}
\begin{proof}
The first part follows from the Markov chain Law of Large Numbers.
For the second part, (\ref{eq:avarjumppif}) follows from the definition
of ${\rm var}(f/\varrho,\tilde{P})$. For (\ref{eq:avarjumpGeo}),
we apply the law of total variance
\[
n{\rm var}\left(Y_{n}^{{\rm Geo}}(f)\right)=n\mathbb{E}\left[{\rm var}\left(Y_{n}^{{\rm Geo}}(f)\mid\tilde{X}\right)\right]+n{\rm var}\left(\mathbb{E}\left[Y_{n}^{{\rm Geo}}\mid\tilde{X}\right]\right),
\]
and observe that for any $n\in\mathbb{N}$,
\begin{eqnarray*}
n\mathbb{E}\left[{\rm var}\left(Y_{n}^{{\rm Geo}}(f)\mid\tilde{X}\right)\right] & = & n\pi(\varrho)^{2}\mathbb{E}\left[{\rm var}\left(\frac{1}{n}\sum_{i=1}^{n}\tau_{i}f(\tilde{X}_{i})\mid\tilde{X}\right)\right]\\
 & = & n\pi(\varrho)^{2}\mathbb{E}\left[\frac{1}{n^{2}}\sum_{i=1}^{n}f(\tilde{X}_{i})^{2}{\rm var}\left(\tau_{i}\mid\tilde{X}_{i}\right)\right]\\
 & = & \pi(\varrho)^{2}\mathbb{E}\left[\frac{1}{n}\sum_{i=1}^{n}f(\tilde{X}_{i})^{2}\frac{1-\varrho(\tilde{X}_{i})}{\varrho(\tilde{X}_{i})^{2}}\right]\\
 & = & \pi(\varrho)^{2}\tilde{\pi}\left(f^{2}\cdot(1-\varrho)/\varrho^{2}\right)\\
 & = & \pi(\varrho)\pi\left(f^{2}\cdot(1-\varrho)/\varrho\right),
\end{eqnarray*}
while
\begin{eqnarray*}
n{\rm var}\left(\mathbb{E}\left[Y_{n}^{{\rm Geo}}(f)\mid\tilde{X}\right]\right) & = & n{\rm var}\left(\mathbb{E}\left[\frac{\pi(\varrho)}{n}\sum_{i=1}^{n}\tau_{i}f(\tilde{X}_{i})\mid\tilde{X}\right]\right)\\
 & = & n{\rm var}\left(\frac{\pi(\varrho)}{n}\sum_{i=1}^{n}f(\tilde{X}_{i})/\varrho(\tilde{X}_{i})\right)\\
 & = & n{\rm var}(Y_{n}^{{\rm RB}}(f)),
\end{eqnarray*}
and the result follows from (\ref{eq:avarjumppif}) by taking the
limit $n\rightarrow\infty$.
\end{proof}
\begin{proof}[Proof of Proposition~\ref{prop:avarrbgeosn}]
The first part follows from the Markov chain Law of Large Numbers
applied to $\frac{1}{n}\sum_{i=1}^{n}f(\tilde{X}_{i})/\varrho(\tilde{X}_{i})$,
$\frac{1}{n}\sum_{i=1}^{n}1/\varrho(\tilde{X}_{i})$, $\frac{1}{n}\sum_{i=1}^{n}\tau_{i}f(\tilde{X}_{i})$
and $\frac{1}{n}\sum_{i=1}^{n}\tau_{i}$. The second part follows
from 
\begin{eqnarray*}
\sqrt{n}\left[\bar{Y}_{n}^{{\rm RB}}(f)-\pi(f)\right] & = & \sqrt{n}\left[\frac{\sum_{i=1}^{n}f(\tilde{X}_{i})/\varrho(\tilde{X}_{i})}{\sum_{i=1}^{n}1/\varrho(\tilde{X}_{i})}-\pi(f)\right]\\
 & = & \sqrt{n}\left[\frac{\sum_{i=1}^{n}\bar{f}(\tilde{X}_{i})/\varrho(\tilde{X}_{i})}{\sum_{i=1}^{n}1/\varrho(\tilde{X}_{i})}\right]\\
 & = & \frac{\frac{1}{\sqrt{n}}\pi(\varrho)\sum_{i=1}^{n}\bar{f}(\tilde{X}_{i})/\varrho(\tilde{X}_{i})}{\pi(\varrho)\frac{1}{n}\sum_{i=1}^{n}1/\varrho(\tilde{X}_{i})},
\end{eqnarray*}
where the denominator converges almost surely to $1$ by the Markov
chain Law of Large Numbers and the numerator converges weakly to a
mean $0$ normal random variable with variance $\sigma_{{\rm RB}}^{2}(f)$
by Lemma~\ref{lem:avarRBGeo} and \citet[Corollary~6]{haggstrom2007variance};
the result follows from Slutsky's lemma. For the third part, similar
to the second part we obtain 
\[
\sqrt{n}\left[\bar{Y}_{n}^{{\rm Geo}}(f)-\pi(f)\right]=\frac{\frac{1}{\sqrt{n}}\pi(\varrho)\sum_{i=1}^{n}\tau_{i}\bar{f}(\tilde{X}_{i})}{\pi(\varrho)\frac{1}{n}\sum_{i=1}^{n}\tau_{i}},
\]
where the denominator converges almost surely to $1$ by the Markov
chain Law of Large Numbers and the numerator converges weakly to a
mean $0$ normal random variable with variance $\sigma_{{\rm Geo}}^{2}(f)$
by Lemma~\ref{lem:avarRBGeo} and \citet[Corollary~6]{haggstrom2007variance};
the result follows from Slutsky's lemma.
\end{proof}
\bibliographystyle{abbrvnat}
\bibliography{whichaverages}

\end{document}